\documentclass[journal,twocolumn]{IEEEtran}
\usepackage{amssymb}
\usepackage{amsmath}
\usepackage{amsfonts}
\usepackage{cite}

\newtheorem{theorem}{Theorem}

\newtheorem{example}{Example}

\newtheorem{lemma}{Lemma}

\newtheorem{proposition}{Proposition}
\newtheorem{remark}{Remark}

\input{tcilatex}
\usepackage{graphicx}
\begin{document}

\title{Epidemic Spread in Human Networks}
\author{Faryad Darabi Sahneh and Caterina Scoglio \\
Electrical \& Computer Engineering Department, Kansas State University,
Manhattan, KS 66506, USA\\
email: \{faryad,caterina\}@ksu.edu}
\maketitle

\begin{abstract}
One of the popular dynamics on complex networks is the epidemic spreading.
An epidemic model describes how infections spread throughout a network.
Among the compartmental models used to describe epidemics, the
Susceptible-Infected-Susceptible (SIS) model has been widely used. In the
SIS model, each node can be susceptible, become infected with a given
infection rate, and become again susceptible with a given curing rate. In
this paper, we add a new compartment to the classic SIS model to account for
human response to epidemic spread. Each individual can be infected,
susceptible, or alert. Susceptible individuals can become alert with an
alerting rate if infected individuals exist in their neighborhood. An
individual in the alert state is less probable to become infected than an
individual in the susceptible state; due to a newly adopted cautious
behavior. The problem is formulated as a continuous-time Markov process on a
general static graph and then modeled into a set of ordinary differential
equations using mean field approximation method and the corresponding
Kolmogorov forward equations. The model is then studied using results from
algebraic graph theory and center manifold theorem. We analytically show
that our model exhibits two distinct thresholds in the dynamics of epidemic
spread. Below the first threshold, infection dies out exponentially. Beyond
the second threshold, infection persists in the steady state. Between the
two thresholds, the infection spreads at the first stage but then dies out
asymptotically as the result of increased alertness in the network. Finally,
simulations are provided to support our findings. Our results suggest that
alertness can be considered as a strategy of controlling the epidemics which
propose multiple potential areas of applications, from infectious diseases
mitigations to malware impact reduction.
\end{abstract}

\section{Introduction}

Modeling human reactions to the spread of infectious disease is an important
topic in current epidemiology \cite{ferguson2007Nature,Funk2010JRSI}, and
has recently attracted a substantial attention \cite%
{theodorakopoulos2011ACC,kitchovitch2010ICCS,zeng2002JAMIA,bauch2004NAS,chen2006JMB,funk2009NAS,funk2010JTB,kiss2010MB}%
. However, few papers are available in the literature which consider the
human response to the epidemic in a systematic framework and the
contributions to the problem are still in an early stage. The challenges in
this topic concern not only how to model human reactions to the presence of
epidemics, but also how these reactions affect the spread of the disease
itself. In a general view, human response to an epidemic spread can be
categorized in the following three types: 1) Change in the system state. For
example, in a vaccination scenario individuals go directly from susceptible
state to recovered without going through infected state. 2) Change in system
parameters. For example, as in \cite{tracht2010PloS}, individuals might
choose to use masks. Those who use masks have a smaller infection rate
parameter, 3) Change in the contact topology. For example, due to the
perception of a serious danger, individuals reduce their contacts with other
people who can potentially be infectious \cite{Funk2010JRSI}.

Early results on epidemic modeling dates back to \cite{McKendrick1925EMS}.
In \cite{bailey1975Book} an epidemic model on a homogenous network was
studied. Later on, results for heterogeneous networks were reported in \cite%
{Vespignani2002EPJB}. Pastor-Satorras \emph{et. al.} \cite{Vespignani2001PRE}
studied epidemic spreading in scale free networks, showing that in these
networks the epidemic threshold vanishes with consequent concerns for the
robustness of many real complex systems. Wang \emph{et. al.} \cite%
{wang2003SRDS} provided the first result for a non-synthetic contact
topology, and studied the epidemic spread dynamic on a general static graph.
Through a local analysis of a mean-field discrete model, it was shown that
the epidemic threshold is directly related to the inverse of the spectral
radius of the adjacency matrix of the contact graph. More detailed proof was
provided in \cite{chakrabarti2008TISSEC}. Ganash \emph{et. al.} \cite%
{Ganesh2005INFOCOM} proved the same result without any mean-field
approximations. A continuous-time epidemic model was studied by Van Mieghem 
\emph{et. al.} \cite{van2009TN}, where a set of ordinary differential
equations was extracted through mean-field approximation of a continuous
time Markov process. The relation between the epidemic threshold and the
spectral radius was rigorously proved and further insights about the steady
state infection probabilities were analytically derived. Preciado and
Jadbabaie \cite{preciado2010CDC} studied the epidemic spread on geometric
random networks and then in \cite{preciado2010arXiv}, they investigated the
epidemic threshold on a general contact graph with respect to the network
structural information.

A good review on existing results in the literature where the human behavior
is taken into account for epidemic modeling can be found in \cite%
{Funk2010JRSI}. Poletti \emph{et. al. }\cite{poletti2009JTB} developed a
population-based model where susceptible individuals could choose between
two behaviors in response to presence of infection. Funk \emph{et. al.} \cite%
{funk2009NAS} showed that awareness of individuals about the presence of a
disease can help reducing the size of the epidemic outbreak. In their paper,
awareness and disease have interconnected dynamics. Theodorakopoulos \emph{%
et. al.} \cite{theodorakopoulos2011ACC} formulated the problem so that
individuals could make decision based on the perception of the epidemic
size. Most of the existing results are suitable for a society of well-mixed
individuals, since the contact graph is usually considered to be homogeneous
(i.e. all nodes have the same degree). To the authors' knowledge, the study
of the human response in a realistic network of individuals with a general
contact graph has not been reported so far.

In this paper, we model the human response to epidemic in the following way.
A new compartment is considered in addition to susceptible and infected
states. A susceptible individual becomes alert with some probability rate if
surrounded by infected individuals. An alert node gets infected with a lower
rate compared to a susceptible node does with the same number of infected
neighbors. The contribution of this paper is two-fold. 1) Unlike most of the
previous results, no homogeneity assumption is made on the contact network
and the human-disease interaction in this paper is modeled on a general
contact graph. 2) We show through analytical approaches that two distinct
thresholds exist. The two are explicitly computed. To the authors' knowledge
the existence of two distinct thresholds is reported for the first time in
this paper, providing a fundamental progress on previous results.
Additionally, this result has the potential to be applied to mitigate
epidemics in several different complex systems, from human and animal
infectious diseases, to malware propagation in computer and sensor networks.

The rest of the paper is organized as follows. In Section \ref%
{Sec-Background}, some backgrounds on graph theory, center manifold method,
and the N-Intertwined SIS model (developed in \cite{van2009TN}) are
recalled. Section \ref{Sec-Problem Statement} is devoted to the problem
formulation and model derivations. Stability analysis results of the model
are provided in Section \ref{Sec-Main}. Finally, results are examined
through numerical simulations in Section \ref{Sec-Simulation Results}.

\section{Preliminarily and Background\label{Sec-Background}}

\subsection{Graph Theory}

Graph theory (see \cite{Diestel97Book}\textbf{) }is widely used for
representing the contact topology in an epidemic network. Let $\mathcal{G}%
=\left\{ \mathcal{V},\mathcal{E}\right\} $ represent a directed graph, and $%
\mathcal{V=}\left\{ 1,...,N\right\} $\ denote the set of vertices. Every
individual is represented by a vertex. The set of edges is denoted as $%
\mathcal{E\subset V\times V}$. An edge is an ordered pair $(i,j)\in \mathcal{%
E}$ if individual $j$ can be directly infected from individual $i$. In this
paper, we assume that there is no self loop in the graph, that is, $%
(i,i)\notin \mathcal{E}$. $\mathcal{N}_{i}=\left\{ j\in \mathcal{V\mid }%
(j,i)\in \mathcal{E}\right\} $ denotes the neighborhood set of vertex $i$.
Graph $\mathcal{G}$\ is said to be undirected if for any edge $(i,j)\in 
\mathcal{E}$, edge $(j,i)\in \mathcal{E}$. A path is referred by the
sequence of its vertices. A path $\mathcal{P}$ of length $k$ between $v_{0}$%
, $v_{k}$\ is the sequence $\left\{ v_{0},...,v_{k}\right\} $ where $%
(v_{i-1},v_{i})\in \mathcal{E}$ for $i=1,...,k$. Directed graph $\mathcal{G}$%
\ is strongly connected if any two vertices are linked with a path in $%
\mathcal{G}$. $\mathcal{A=}\left[ a_{ij}\right] \in 
\mathbb{R}
^{N\times N}$ denotes the adjacency matrix of $\mathcal{G}$, where $a_{ij}=1$
if and only if $(i,j)\in \mathcal{E}$ else $a_{ij}=0$. The largest magnitude
of the eigenvalues of adjacency matrix $A$ is called spectral radius of $A$
and is denoted by $\rho (A)$.

\subsection{Center Manifold Theory\label{Sec: CMT}}

Linearization is a useful technique for local stability analysis of
nonlinear systems. However, in the cases where linearization results in a
linear system with some negative real part and some zero real part
eigenvalues, the linearization method fails. In these cases, the local
stability analysis can be performed by analyzing a nonlinear system of the
order exactly equal to the number of eigenvalues with zero real parts. This
method is known as center manifold method. In this section, we have a quick
review on center manifold theory. More details can be found in \cite%
{Carr1981Book} and \cite{khalil2002Book}.

For $z_{s}\in 
\mathbb{R}
^{n_{s}}$ and $z_{c}\in 
\mathbb{R}
^{n_{c}}$, consider the following system%
\begin{eqnarray}
\dot{z}_{s} &=&A_{s}z_{s}+g_{s}(z_{c},z_{s})  \label{zc} \\
\dot{z}_{c} &=&A_{c}z_{c}+g_{c}(z_{c},z_{s}),  \label{zs}
\end{eqnarray}%
where the eigenvalues of $A_{s}\in 
\mathbb{R}
^{n_{s}\times n_{s}}$ and $A_{c}\in 
\mathbb{R}
^{n_{c}\times n_{c}}$ have negative and zero real parts, respectively. The
functions $g_{c}$ and $g_{s}$ are twice continuously differentiable and
satisfy the conditions%
\begin{equation}
g_{i}(\mathbf{0},\mathbf{0})=\mathbf{0},~\nabla g_{i}(\mathbf{0},\mathbf{0})=%
\mathbf{0},~i\in \{s,c\},
\end{equation}%
where $\mathbf{0}$ is a vector or matrix of zeros with appropriate
dimensions. There exists a function $h:%
\mathbb{R}
^{n_{s}}\rightarrow 
\mathbb{R}
^{n_{c}}$ satisfying%
\begin{equation}
h(\mathbf{0})=\mathbf{0},~\nabla h(\mathbf{0})=\mathbf{0},
\end{equation}%
that $z_{s}=h(z_{c})$ is an invariant manifold (see \cite{khalil2002Book}
for the definition) for (\ref{zc}) and (\ref{zs}) near the origin. The
dynamic system (\ref{zc}) and (\ref{zs}) can be studied through the reduced
system%
\begin{equation}
\overset{\cdot }{\hat{z}}_{c}=A_{c}\hat{z}_{c}+g_{c}(\hat{z}_{c},h(\hat{z}%
_{c})).
\end{equation}%
The invariant manifold $z_{s}=h(z_{c})$ is a center manifold for the system (%
\ref{zc}) and (\ref{zs}), i.e., every trajectory of (\ref{zc}) and (\ref{zs}%
) with the initial condition $z_{c}(0)=\hat{z}_{c}(0)$ and $z_{s}(0)=h(\hat{z%
}_{c}(0))$ satisfies $z_{c}(t)=\hat{z}_{c}(t)$ and $z_{s}(t)=h(\hat{z}%
_{c}(t))$. In addition, small deviation from the center manifold is
exponentially attracted, i.e., if $\left\Vert z_{s}(0)-h(\hat{z}%
_{c}(0))\right\Vert $ is small enough, then $\left\Vert z_{s}(t)-h(\hat{z}%
_{c}(t))\right\Vert $ will go to zero exponentially.

\subsection{N-Intertwined SIS Model for Epidemic Spread\label{Sec: NI_Model}}

We have built our modeling based on a newly proposed continuous-time model
for epidemic spread on a graph. Van Mieghem \emph{et. al.} \cite{van2009TN}
derived a set of ordinary differential equations, called the N-intertwined
model, which represents the time evolution of the probability of infection
for each individual. The only approximation for the N-intertwined model
corresponds to the application of the mean-field theory.

Consider a network of $N$ individuals. Denote the infection probability of
the $i$-th individual by $p_{i}\in \lbrack 0,1]$. Assume that the disease is
characterized by infection rate $\beta _{0}\in 
\mathbb{R}
^{+}$ and cure rate $\delta \in 
\mathbb{R}
^{+}$. Furthermore, assume that the contact topology is represented by a
static graph. The N-intertwined model proposed in \cite{van2009TN}\ is

\begin{equation}
\dot{p}_{i}=\beta _{0}(1-p_{i})\sum_{j\in \mathcal{N}_{i}}a_{ij}p_{j}-\delta
p_{i},~i\in \{1,...,N\},  \label{NInt_Model}
\end{equation}%
where $a_{ij}=1$ if individual $j$ is a neighbor of individual $i$,
otherwise $a_{ij}=0$.

\begin{proposition}
\label{Prop: Die out}Consider the N-intertwined model (\ref{NInt_Model}).
Initial infection will die out exponentially if the infection strength $\tau
\triangleq \frac{\beta _{0}}{\delta }$ satisfies%
\begin{equation}
\tau \triangleq \frac{\beta _{0}}{\delta }\leq \frac{1}{\rho (A)},
\end{equation}%
where $\rho (A)$ is the spectral radius of the adjacency matrix $A$ of the
contact graph.
\end{proposition}

\begin{remark}
The value $\tau _{c}=\frac{1}{\rho (A)}$ is usually referred to as the
epidemic threshold. For any infection strength $\tau >\tau _{c}$, infection
will persist in the steady state. The following result discusses the steady
state values for infection probabilities.
\end{remark}

\begin{proposition}
\label{Prop: NI_SS}If the infection strength is above the epidemic
threshold, the steady state values of the infection probabilities, denoted
by $p_{i}^{ss}$ for the $i$-th individual, is the non-trivial solution of
the following set of equations%
\begin{equation}
\frac{\beta _{0}}{\delta }\sum_{j\in \mathcal{N}_{i}}a_{ij}p_{j}^{ss}=\frac{%
p_{i}^{ss}}{1-p_{i}^{ss}},~~i\in \{1,...,N\}.  \label{NI_pss}
\end{equation}
\end{proposition}

\section{Model Development\label{Sec-Problem Statement}}

In this paper, we add a new compartment to the classic SIS\ model for
epidemic spread modeling to propose a Susceptible-Alert-Infected-Susceptible
(SAIS) model. The contact topology in this formulation is considered as a
general static graph. Each node of the graph represents an individual and a
link between two nodes determines the contact between the two individuals.
Each node is allowed to be in one of the three states \emph{"S: susceptible"}%
, \emph{"I: infected"}, and \emph{"A: alert"}. A susceptible individual
becomes infected by the infection rate $\beta _{0}$ times the number of its
infected neighbors. An infected individual recovers back to the susceptible
state by the curing rate $\delta $. An individual can observe the states of
its neighbors. A susceptible individual might go to the alert state if
surrounded by infected individuals. Specifically, a susceptible node becomes
alert with the alerting rate $\kappa \in 
\mathbb{R}
^{+}$ times the number of infected neighbors. An alert individual can get
infected in a process similar to a susceptible individual but with a reduced
infection rate $0\leq \beta _{a}<\beta _{0}$. We assume that transition from
an alert individual to a susceptible state is much slower than other
transitions. Hence, in our modeling setup, an alert individual never goes
directly to the susceptible state. The compartmental transitions of a node
with one single infected neighbor are depicted in Fig. \ref{sais.emf}.

\begin{figure}[h!]
\begin{center}
\leavevmode
\includegraphics[scale=0.4, trim=0cm 0cm 0cm 1.5cm, clip]{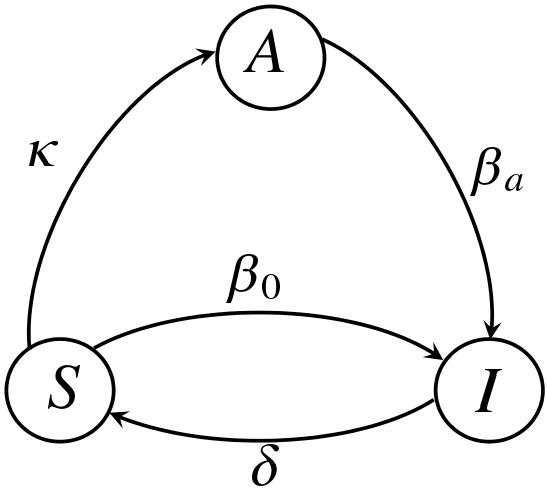}
\end{center}
\caption{The infected population fraction in Example. (a) SIS model. (b) SAIS model with. (c) SAIS model with . The Monte-Carlo simulation results are shown in blue.}
\label{sais.emf}
\end{figure}

The epidemic spread dynamic is modeled as a continuous-time Markov process.
For each node $i\in \{1,...,N\}$, define a random variable $%
X_{i}:\{S,I,A\}\rightarrow \{0,1,2\}$. Denote $X_{i}^{t}$ a measure of the
random variable at time $t$ for node $i$. The epidemic spread dynamics is
modeled as the following continuous-time Markov process:%
\begin{align}
\Pr [X_{i}^{t+\Delta t}& =1|X_{i}^{t}=0]=\beta _{0}\Delta t\sum_{j\in 
\mathcal{N}_{i}}1_{\{X_{j}^{t}=1\}}+o(\Delta t),  \notag \\
\Pr [X_{i}^{t+\Delta t}& =0|X_{i}^{t}=1]=\delta \Delta t+o(\Delta t),  \notag
\\
\Pr [X_{i}^{t+\Delta t}& =2|X_{i}^{t}=0]=\kappa \Delta t\sum_{j\in \mathcal{N%
}_{i}}1_{\{X_{j}^{t}=1\}}+o(\Delta t),  \notag \\
\Pr [X_{i}^{t+\Delta t}& =1|X_{i}^{t}=2]=\beta _{a}\Delta t\sum_{j\in 
\mathcal{N}_{i}}1_{\{X_{j}^{t}=1\}}+o(\Delta t),  \label{Makov1_4}
\end{align}%
for $i\in \{1,...,N\}$. In (\ref{Makov1_4}), $\Pr [\cdot ]$ denotes
probability, $\Delta t>0$ is a time step, and $1_{\{\mathcal{X}\}}$ is one
if $\mathcal{X}$\ is true and zero otherwise. A function $f(\Delta t)$ is
said to be $o(\Delta t)$ if $\lim_{\Delta t\rightarrow 0}\frac{f(\Delta t)}{%
\Delta t}=0$.

A common approach for studying a continuous-time Markov process is to derive
the corresponding Kolmogorov forward (backward) differential equations (see 
\cite{stroock2005Book} and \cite{Piet2006Book}). As can be seen from the
above equations, the conditional transition probabilities of a node are
expressed in terms of the actual state of its neighboring nodes. Therefore,
each state of the Kolmogorov differential equations corresponding to the
Markov process (\ref{Makov1_4}) will be the probability of being in a
specific configuration. In this case, we will end up with a set of first
order ordinary differential equations of the order $3^{N}$. Hence, the
analysis will become dramatically complicated as the network size grows. In
addition, it is more desirable to study the probability that each individual
is susceptible, infected, or alert. Using a proper mean-field approximation,
it is possible to express the transition probabilities in terms of infection
probabilities of the neighbors. Specifically, the term $1_{\{X_{j}^{t}=1\}}$
is replaced with $\Pr [X_{j}^{t}=1]$ in (\ref{Makov1_4}). Hence, the
following new stochastic process is obtained:%
\begin{align}
\Pr [X_{i}^{t+\Delta t}& =1|X_{i}^{t}=0]=\beta _{0}\Delta t\sum_{j\in 
\mathcal{N}_{i}}\Pr [X_{j}^{t}=1]+o(\Delta t),  \notag \\
\Pr [X_{i}^{t+\Delta t}& =0|X_{i}^{t}=1]=\delta \Delta t+o(\Delta t),  \notag
\\
\Pr [X_{i}^{t+\Delta t}& =2|X_{i}^{t}=0]=\kappa \Delta t\sum_{j\in \mathcal{N%
}_{i}}\Pr [X_{j}^{t}=1]+o(\Delta t),  \notag \\
\Pr [X_{i}^{t+\Delta t}& =1|X_{i}^{t}=2]=\beta _{a}\Delta t\sum_{j\in 
\mathcal{N}_{i}}\Pr [X_{j}^{t}=1]+o(\Delta t).  \label{Markov_MF}
\end{align}

Define a new state $x_{i}\triangleq \lbrack s_{i},p_{i},q_{i}]^{T}$, where $%
s_{i}$, $p_{i}$, and $q_{i}$ denote the probabilities of individual $i$ to
be susceptible, infected, and alert, respectively. The Kolmogorov forward
differential equations of the stochastic process (\ref{Markov_MF}) can now
be found as%
\begin{equation}
\dot{x}_{i}=\Theta _{i}^{T}x_{i},~i\in \{1,...,N\},  \label{Kolmogorov}
\end{equation}%
where%
\begin{equation}
\Theta _{i}\triangleq 
\begin{bmatrix}
-\delta & 0 & \delta \\ 
\beta _{a}\sum\limits_{j\in \mathcal{N}_{i}}a_{ij}p_{j} & -\beta
_{a}\sum\limits_{j\in \mathcal{N}_{i}}a_{ij}p_{j} & 0 \\ 
\beta _{0}\sum\limits_{j\in \mathcal{N}_{i}}a_{ij}p_{j} & \kappa
\sum\limits_{j\in \mathcal{N}_{i}}a_{ij}p_{j} & -(\beta _{0}+\kappa
)\sum\limits_{j\in \mathcal{N}_{i}}a_{ij}p_{j}%
\end{bmatrix}%
\end{equation}%
is the infinitesimal transition matrix. One property of the dynamic system (%
\ref{Kolmogorov}) is that $s_{i}+p_{i}+q_{i}$ is a preserved quantity.
Hence, the states $s_{i}$, $p_{i}$, and $q_{i}$ are not independent.
Omitting $s_{i}$ in (\ref{Kolmogorov}), the following set of differential
equations is obtained:%
\begin{align}
\dot{p}_{i}& =\beta _{0}(1-p_{i}-q_{i})\sum_{j\in \mathcal{N}%
_{i}}a_{ij}p_{j}+\beta _{a}q_{i}\sum_{j\in \mathcal{N}_{i}}a_{ij}p_{j}-%
\delta p_{i},  \label{dp} \\
\dot{q}_{i}& =\kappa (1-p_{i}-q_{i})\sum_{j\in \mathcal{N}%
_{i}}a_{ij}p_{j}-\beta _{a}q_{i}\sum_{j\in \mathcal{N}_{i}}a_{ij}p_{j},
\label{dq}
\end{align}%
for $i\in \{1,...,N\}$.

\begin{remark}
As can be seen using a mean-field approximation, the dimension of the
differential equations is reduced from $3^{N}$ to $2N$. However, some
information is definitely lost and there is some error. For example, the
Markov process (\ref{Makov1_4}) exhibits an absorbing state. However, no
absorbing state can be observed based on the equations (\ref{dp}) and (\ref%
{dq}). In addition, as is discussed in \cite{van2009TN}, the solution from
the mean field approximation is an upper-bound for the actual model.
\end{remark}

\section{Behavioral Study of SAIS Epidemic Spread Model\label{Sec-Main}}

In this section, the dynamic system (\ref{dp}) and (\ref{dq}) derived in the
previous section is analyzed. It is shown that alertness decreases the size
of infection. In addition, in an SAIS\ epidemic model, the response of the
system can be categorized in three separate regions. These three regions are
identified with two distinct thresholds $\tau _{c}^{1}$ and $\tau _{c}^{2}$.
Below the first threshold, the epidemic dies out exponentially. Beyond the
second threshold, the epidemic persists in the steady state. Between $\tau
_{c}^{1}$ and $\tau _{c}^{2}$, the epidemic spreads at the first stage but
then dies out asymptotically as the result of increased alertness in the
network.

\subsection{Comparison between SAIS and SIS}

In this section, the SAIS model and the SIS model are compared in the sense
of infection probabilities of the individuals. Specifically, we are
interested to compare $p_{i}(t)$, the response of (\ref{dp}) and (\ref{dq}),
with infection probability $p_{i}^{\prime }(t)$ in the N-intertwined SIS
model, which is the solution of the system%
\begin{equation}
\dot{p}_{i}^{\prime }=\beta _{0}(1-p_{i}^{\prime })\sum_{j\in \mathcal{N}%
_{i}}a_{ij}p_{j}^{\prime }-\delta p_{i}^{\prime }.  \label{pidot}
\end{equation}%
It is shown that alertness decreases the probability of infection for each
individual. This result is stated as the following theorem.

\begin{theorem}
\label{Theorem: upperbound}Starting with the same initial conditions $%
p_{i}(t_{0})=p_{i}^{\prime }(t_{0})$, $i=\{1,...,N\}$, the infection
probabilities of individuals\ in SIS\ model (\ref{pidot}) always dominate
those of the SAIS\ model (\ref{dp}) and (\ref{dq}), i.e.,%
\begin{equation}
p_{i}(t)\leq p_{i}^{\prime }(t),i=\{1,...,N\}~~\forall t\in \lbrack
t_{0},\infty ).
\end{equation}
\end{theorem}

\begin{proof}
Rewrite the equations (\ref{dp}) as%
\begin{equation}
\dot{p}_{i}=\beta _{0}(1-p_{i})\sum_{j\in \mathcal{N}_{i}}a_{ij}p_{j}-(\beta
_{0}-\beta _{a})q_{i}\sum_{j\in \mathcal{N}_{i}}a_{ij}p_{j}-\delta p_{i}.
\label{p_dot2}
\end{equation}%
Starting with the same initial conditions $p_{i}(t_{0})=p_{i}^{\prime
}(t_{0})$, it is concluded that%
\begin{equation}
p_{i}(t_{0})=p_{i}^{\prime }(t_{0})\Rightarrow \dot{p}_{i}(t_{0})\leq \dot{p}%
_{i}^{\prime }(t_{0}),  \label{pi_pidot_to}
\end{equation}%
since $\beta _{a}<\beta _{0}$ by definition and therefore $(\beta _{0}-\beta
_{a})q_{i}(t_{0})\sum_{j\in \mathcal{N}_{i}}a_{ij}p_{j}(t_{0})$ is a
non-negative term. According to (\ref{pi_pidot_to}), there exists $%
t_{f}>t_{0}$ so that%
\begin{equation}
p_{i}(t)\leq p_{i}^{\prime }(t),i\in \{1,...,N\}~~\forall t\in \lbrack
t_{0},t_{f}].  \label{pi_ineq_tf}
\end{equation}%
The theorem is proved if we show that inequality (\ref{pi_ineq_tf}) holds
for every $t_{f}\in (t_{0},\infty )$. Assume that there exists $t_{1}>t_{0}$%
, so that (\ref{pi_ineq_tf}) holds for $t_{f}=t_{1}$ but it is not true for
any $t_{f}>t_{1}$. Obviously, at $t=t_{1}$,%
\begin{equation}
\exists i\in \{1,...,N\}\text{ so that }p_{i}(t_{1})=p_{i}^{\prime }(t_{1})%
\text{ and }\dot{p}_{i}(t_{1})>\dot{p}_{i}^{\prime }(t_{1}).
\label{pi_contradiction}
\end{equation}%
In the subsequent arguments, it is shown that no such $t_{1}$ exists. From (%
\ref{p_dot2}), $\dot{p}_{i}(t_{1})$ is found to satisfy%
\begin{eqnarray}
\dot{p}_{i}(t_{1}) &=&\beta _{0}(1-p_{i}(t_{1}))\sum_{j\in \mathcal{N}%
_{i}}a_{ij}p_{j}(t_{1})  \notag \\
&&-(\beta _{0}-\beta _{a})q_{i}(t_{1})\sum_{j\in \mathcal{N}%
_{i}}a_{ij}p_{j}(t_{1})-\delta p_{i}(t_{1})  \notag \\
&\leq &\beta _{0}(1-p_{i}(t_{1}))\sum_{j\in \mathcal{N}%
_{i}}a_{ij}p_{j}(t_{1})-\delta p_{i}(t_{1})  \notag \\
&=&\beta _{0}(1-p_{i}^{\prime }(t_{1}))\sum_{j\in \mathcal{N}%
_{i}}a_{ij}p_{j}(t_{1})-\delta p_{i}^{\prime }(t_{1}),
\label{pidot_ineq_t1_1}
\end{eqnarray}%
according to (\ref{pi_contradiction}) and the fact that $(\beta _{0}-\beta
_{a})q_{i}(t_{1})\sum a_{ij}p_{j}(t_{1})$\ is a non-negative term. Based on (%
\ref{pi_ineq_tf}), $\forall j\in \{1,...,N\}$ we have $p_{j}(t_{1})\leq
p_{j}^{\prime }(t_{1})$. Therefore, the inequality (\ref{pidot_ineq_t1_1})
is further simplified as%
\begin{equation}
\dot{p}_{i}(t_{1})\leq \beta _{0}(1-p_{i}^{\prime }(t_{1}))\sum_{j\in 
\mathcal{N}_{i}}a_{ij}p_{j}^{\prime }(t_{1})-\delta p_{i}^{\prime }(t_{1})=%
\dot{p}_{i}^{\prime }(t_{1}).  \label{pidot_ineq_t1_2}
\end{equation}%
Having $\dot{p}_{i}(t_{1})\leq \dot{p}_{i}^{\prime }(t_{1})$ contradicts
with (\ref{pi_contradiction}). Therefore, no such $t_{1}$ exists so that (%
\ref{pi_contradiction}) is true. As a result the inequality (\ref{pi_ineq_tf}%
) holds for every $t_{f}\in (t_{0},\infty )$. This completes the proof.
\end{proof}

\subsection{Exponential Epidemic Die-Out\label{Exp_dieout}}

\begin{theorem}
\label{Theorem: ES}Consider the SAIS epidemic spread model (\ref{dp}) and (%
\ref{dq}). Assume that the infection strength satisfies%
\begin{equation}
\tau =\frac{\beta _{0}}{\delta }<\frac{1}{\rho (A)}.  \label{Threshold1}
\end{equation}%
Then, initial infections will die out exponentially.
\end{theorem}

\begin{proof}
The solution of $p_{i}(t)$ was proved in Theorem \ref{Theorem: upperbound}
to be upper-bounded by $p_{i}^{\prime }(t)$. As described in Section \ref%
{Sec: NI_Model} and based on Proposition \ref{Prop: Die out}, the
N-intertwined model (\ref{pidot}) is exponentially stable if (\ref%
{Threshold1}) is satisfied. As a consequence, $p_{i}(t)$ in (\ref{p_dot2})
is also exponentially stable if (\ref{Threshold1}) is satisfied.
\end{proof}

\begin{remark}
In the proof of Theorem \ref{Theorem: upperbound}, $q_{i}$ is considered as
a non-negative time-varying term. Under the conditions of Theorem \ref%
{Theorem: ES}, $q_{i}(t)$ will regulate at some value depending on the
initial conditions.
\end{remark}

\begin{remark}
Note that adding the alert compartment does not contribute to the epidemic
threshold for exponential die out. This result is already concluded in \cite%
{funk2009NAS} for a homogeneous network (i.e. all nodes have the same
degree).
\end{remark}

\subsection{Asymptotically Epidemic Die-Out}

According to (\ref{dq}),

\begin{equation}
q_{i}^{e}=\frac{1-p_{i}}{1+\frac{\beta _{a}}{\kappa }},~i\in \{1,...,N\},
\end{equation}%
is an equilibrium for (\ref{dq}). To facilitate the subsequent analysis,
define a new state $r_{i}$ as%
\begin{equation}
r_{i}\triangleq q_{i}-q_{i}^{e}=q_{i}-\frac{1-p_{i}}{1+\frac{\beta _{a}}{%
\kappa }}.  \label{ri}
\end{equation}%
The derivatives $\dot{p}_{i}$ and $\dot{r}_{i}$ in the new coordinate can be
found by substituting $q_{i}=r_{i}+\frac{1}{1+\frac{\beta _{a}}{\kappa }}-%
\frac{p_{i}}{1+\frac{\beta _{a}}{\kappa }}$ from (\ref{ri}) in (\ref{dp})
and (\ref{dq}) as%
\begin{eqnarray}
\dot{p}_{i} &=&\beta _{0}(1-p_{i}-\{r_{i}+\frac{1}{1+\frac{\beta _{a}}{%
\kappa }}-\frac{p_{i}}{1+\frac{\beta _{a}}{\kappa }}\})\sum_{j\in \mathcal{N}%
_{i}}a_{ij}p_{j}  \notag \\
&&+\beta _{a}\{r_{i}+\frac{1}{1+\frac{\beta _{a}}{\kappa }}-\frac{p_{i}}{1+%
\frac{\beta _{a}}{\kappa }}\}\sum_{j\in \mathcal{N}_{i}}a_{ij}p_{j}-\delta
p_{i}  \notag \\
&=&\{\beta _{0}\frac{\frac{\beta _{a}}{\kappa }}{1+\frac{\beta _{a}}{\kappa }%
}+\beta _{a}\frac{1}{1+\frac{\beta _{a}}{\kappa }}\}\sum_{j\in \mathcal{N}%
_{i}}a_{ij}p_{j}  \notag \\
&&-\{\beta _{0}+\frac{\beta _{0}+\beta _{a}}{1+\frac{\beta _{a}}{\kappa }}%
\}p_{i}\sum_{j\in \mathcal{N}_{i}}a_{ij}p_{j}  \notag \\
&&-(\beta _{0}-\beta _{a})r_{i}\sum_{j\in \mathcal{N}_{i}}a_{ij}p_{j}-\delta
p_{i},  \label{pdot_cm}
\end{eqnarray}%
and%
\begin{eqnarray}
\dot{r}_{i} &=&\kappa (1-p_{i}-\{r_{i}+\frac{1}{1+\frac{\beta _{a}}{\kappa }}%
-\frac{p_{i}}{1+\frac{\beta _{a}}{\kappa }}\})\sum_{j\in \mathcal{N}%
_{i}}a_{ij}p_{j}  \notag \\
&&-\beta _{a}\{r_{i}+\frac{1}{1+\frac{\beta _{a}}{\kappa }}-\frac{p_{i}}{1+%
\frac{\beta _{a}}{\kappa }}\}\sum_{j\in \mathcal{N}_{i}}a_{ij}p_{j}  \notag
\\
&=&-\kappa (1+\frac{\beta _{a}}{\kappa })r_{i}\sum_{j\in \mathcal{N}%
_{i}}a_{ij}p_{j}.  \label{rdot_cm}
\end{eqnarray}

To facilitate the subsequent analysis, define%
\begin{eqnarray}
\mathbf{p} &\triangleq &[p_{1},...,p_{N}]^{T}\in 
\mathbb{R}
^{N},  \label{p_v} \\
\mathbf{r} &\triangleq &[r_{1},...,r_{N}]^{T}\in 
\mathbb{R}
^{N}.  \label{r_v}
\end{eqnarray}%
According to (\ref{pdot_cm}) and (\ref{rdot_cm}) and the definitions (\ref%
{p_v}) and (\ref{r_v}), the followings are true%
\begin{eqnarray}
\mathbf{\dot{p}} &=&(\beta _{eq}A-\delta I)\mathbf{p}+G_{1}(\mathbf{p},%
\mathbf{r}),  \label{p_v_dot} \\
\mathbf{\dot{r}} &=&\mathbf{0r}+G_{2}(\mathbf{p},\mathbf{r}),
\label{r_v_dot}
\end{eqnarray}%
where%
\begin{equation}
\beta _{eq}\triangleq \beta _{0}\frac{\frac{\beta _{a}}{\kappa }}{1+\frac{%
\beta _{a}}{\kappa }}+\beta _{a}\frac{1}{1+\frac{\beta _{a}}{\kappa }},
\label{bet_eq}
\end{equation}%
and%
\begin{eqnarray}
G_{1}(\cdot ) &\triangleq &[g_{1,1}(\cdot ),...,g_{1,N}(\cdot )]^{T},
\label{G1} \\
G_{2}(\cdot ) &\triangleq &[g_{2,1}(\cdot ),...,g_{2,N}(\cdot )]^{T},
\label{G2}
\end{eqnarray}%
with%
\begin{eqnarray}
g_{1,i}(\mathbf{p},\mathbf{r}) &\triangleq &-\{\beta _{0}+\frac{\beta
_{0}+\beta _{a}}{1+\frac{\beta _{a}}{\kappa }}\}p_{i}\sum_{j\in \mathcal{N}%
_{i}}a_{ij}p_{j}  \notag \\
&&-(\beta _{0}-\beta _{a})r_{i}\sum_{j\in \mathcal{N}_{i}}a_{ij}p_{j}, \\
g_{2,i}(\mathbf{p},\mathbf{r}) &\triangleq &-\kappa (1+\frac{\beta _{a}}{%
\kappa })r_{i}\sum_{j\in \mathcal{N}_{i}}a_{ij}p_{j}.
\end{eqnarray}

If we linearize the system (\ref{p_v_dot}) and (\ref{r_v_dot}), the
resulting system has $N$\ zero eigenvalues. Therefore, linearization
technique fails to investigate the stability properties of (\ref{p_v_dot})
and (\ref{r_v_dot}). In the following arguments, we show that center
manifold theory can be employed to study the stability of (\ref{p_v_dot})
and (\ref{r_v_dot}).

The eigenvalues of matrix $(\beta _{eq}A-\delta I)$ are $\beta _{eq}\lambda
_{i}-\delta ,i\in \{1,...N\}$, where $\lambda _{i}$'s are the eigenvalues of
the adjacency matrix $A$. Therefore, assuming that%
\begin{equation}
\frac{\beta _{eq}}{\delta }<\frac{1}{\rho (A)},  \label{threshold2}
\end{equation}%
the matrix $(\beta _{eq}A-\delta I)$ is Hurwitz (i.e., a matrix that all of
its eigenvalues have negative real parts). In addition, the two nonlinear
functions $G_{1}$ and $G_{2}$ defined in (\ref{G1}) and (\ref{G2}) satisfy%
\begin{equation}
G_{j}(\mathbf{0},\mathbf{0})=\mathbf{0},~\nabla G_{j}(\mathbf{0},\mathbf{0})=%
\mathbf{0,}  \label{Gi_prop}
\end{equation}%
for $j\in \{1,2\}$. Hence, the center manifold theory reviewed in Section %
\ref{Sec: CMT} may apply. The center manifold theorem suggests that there
exists a function $H(\cdot ):%
\mathbb{R}
^{N}\rightarrow 
\mathbb{R}
^{N}$ where the dynamics (\ref{p_v_dot}) and (\ref{r_v_dot}) can be
determined by%
\begin{equation}
\overset{\cdot }{\mathbf{\hat{r}}}=G_{2}(H(\mathbf{\hat{r}}),\mathbf{\hat{r}}%
).  \label{r_v_hat_dot}
\end{equation}

Differential equation (\ref{r_v_hat_dot}) can be written in terms of its
entries as%
\begin{equation}
\overset{\cdot }{\hat{r}}_{i}=-\kappa (1+\frac{\beta _{a}}{\kappa })\hat{r}%
_{i}\sum_{j\in \mathcal{N}_{i}}a_{ij}h_{j}(\mathbf{\hat{r}}),
\label{r_hat_dot}
\end{equation}%
for $i\in \{1,...,N\}$, where $h_{i}(\cdot )$ is the $i$-th component of $%
H(\cdot )\triangleq \lbrack h_{1}(\cdot ),...,h_{N}(\cdot )]^{T}.$

\begin{remark}
\label{Remark: hi_pos}Usually, it is not feasible to find $h_{i}(\cdot )$
explicitly. In the subsequent analysis, instead of explicit calculations, we
make use of the following property of $h_{i}(\cdot )$: Since the probability 
$p_{i}$ is non-negative, each function $h_{i}(\cdot )$ is necessarily
non-negative.
\end{remark}

\begin{lemma}
\label{Lemma: r_hat}The trajectories of (\ref{r_hat_dot}) will
asymptotically converge to the set defined by%
\begin{equation}
\Omega =\{\mathbf{\hat{r}\in 
\mathbb{R}
}^{N}\mathbf{|}\hat{r}_{i}\sum_{j\in \mathcal{N}_{i}}a_{ij}h_{j}(\mathbf{%
\hat{r}})=0\}.
\end{equation}
\end{lemma}

\begin{proof}
Define a continuously differentiable function $V$ as%
\begin{equation}
V\triangleq \frac{1}{2}\mathbf{\hat{r}}^{T}\mathbf{\hat{r}}.
\end{equation}%
Taking the derivative of $V$ with respect to time, we have%
\begin{equation}
\dot{V}=\sum_{i=1}^{N}\hat{r}_{i}\overset{\cdot }{\hat{r}}_{i}=-\kappa (1+%
\frac{\beta _{a}}{\kappa })\sum_{i=1}^{N}\left( \hat{r}_{i}^{2}\sum_{j\in 
\mathcal{N}_{i}}a_{ij}h_{j}(\mathbf{\hat{r}})\right) .
\end{equation}%
It can be seen that the time derivative $\dot{V}$ is negative semi-definite
according to Remark \ref{Remark: hi_pos}. According to the LaSalle's
invariance theorem (see \cite{khalil2002Book}) the trajectories of (\ref%
{r_hat_dot}) will asymptotically converge to the set $\dot{V}\equiv 0$, i.e.,%
\begin{equation}
\Omega \triangleq \{\mathbf{\hat{r}\in 
\mathbb{R}
}^{N}\mathbf{|}\hat{r}_{i}\sum_{j\in \mathcal{N}_{i}}a_{ij}h_{j}(\mathbf{%
\hat{r}})=0\}.
\end{equation}
\end{proof}

\begin{theorem}
\label{Theorem: AS}Consider the SAIS epidemic model (\ref{dp}) and (\ref{dq}%
). Assume that the infection strength satisfies (\ref{threshold2}) where $%
\beta _{eq}$ is defined in (\ref{bet_eq}). Small initial infections die out
asymptotically as $t\rightarrow \infty $.
\end{theorem}

\begin{proof}
Since the infection strength satisfies (\ref{threshold2}), the matrix $%
(\beta _{eq}A-\delta I)$ is Hurwitz. According to the property (\ref{Gi_prop}%
) of $G_{1}(\mathbf{p},\mathbf{r})$, the system%
\begin{equation*}
\mathbf{\dot{p}}=(\beta _{eq}A-\delta I)\mathbf{p}+G_{1}(\mathbf{p},\mathbf{0%
}),
\end{equation*}%
which is system (\ref{p_v_dot}) with $\mathbf{r}=\mathbf{0}$, is
exponentially stable. In addition, according to Lemma \ref{Lemma: r_hat}, $%
\hat{r}_{i}\sum_{j\in \mathcal{N}_{i}}a_{ij}h_{j}(\mathbf{\hat{r}}%
)\rightarrow \infty $ as $t\rightarrow \infty $. Therefore, the term $%
r_{i}\sum a_{ij}p_{j}$ in (\ref{pdot_cm}) can be considered as a decaying
disturbance for (\ref{p_v_dot}). Therefore, $p_{i}\rightarrow 0$
asymptotically as $t\rightarrow \infty $.
\end{proof}

\begin{remark}
\label{two_thresh}From Theorem \ref{Theorem: ES}, the first epidemic
threshold is%
\begin{equation}
\tau _{c}^{1}=\frac{1}{\rho (A)},  \label{taw_1}
\end{equation}%
which is equal to the epidemic threshold in the classic SIS\ epidemic
network. If the infection rate $\beta _{a}$ is such that%
\begin{equation}
\frac{\beta _{a}}{\delta }<\frac{1}{\rho (A)},  \label{bamu}
\end{equation}%
the ratio $\frac{\beta _{eq}}{\delta }$ can be larger or smaller than $\frac{%
1}{\rho (A)}$, depending on the value of $\beta _{0}$. Therefore, if (\ref%
{bamu}) holds, Theorem \ref{Theorem: AS}\ suggests that there exists another
epidemic threshold $\tau _{c}^{2}$. Using the definition of $\beta _{eq}$ in
(\ref{bet_eq}), the condition (\ref{threshold2}) in Theorem \ref{Theorem: AS}%
\ can be expressed as%
\begin{equation}
\frac{\beta _{eq}}{\delta }=\frac{\beta _{0}}{\delta }\frac{\frac{\beta _{a}%
}{\kappa }}{1+\frac{\beta _{a}}{\kappa }}+\frac{\beta _{a}}{\delta }\frac{1}{%
1+\frac{\beta _{a}}{\kappa }}\leq \frac{1}{\rho (A)},
\end{equation}%
which is equivalent to%
\begin{eqnarray}
\frac{\beta _{0}}{\delta } &\leq &\frac{\frac{1}{\rho (A)}-\frac{\beta _{a}}{%
\delta }\frac{1}{1+\frac{\beta _{a}}{\kappa }}}{\frac{\frac{\beta _{a}}{%
\kappa }}{1+\frac{\beta _{a}}{\kappa }}}=\frac{\frac{\beta _{a}}{\kappa }+1}{%
\frac{\beta _{a}}{\kappa }}\frac{1}{\rho (A)}-\frac{\beta _{a}}{\delta }%
\frac{\kappa }{\beta _{a}}  \notag \\
&=&\frac{1}{\rho (A)}+\frac{\kappa }{\beta _{a}}(\frac{1}{\rho (A)}-\frac{%
\beta _{a}}{\delta }).  \label{thresh2}
\end{eqnarray}%
The second epidemic threshold $\tau _{c}^{2}$ can now be obtained from
inequality (\ref{thresh2}) as%
\begin{equation}
\tau _{c}^{2}=\tau _{c}^{1}+\frac{\kappa }{\beta _{a}}(\frac{1}{\rho (A)}-%
\frac{\beta _{a}}{\delta }).  \label{taw_2}
\end{equation}%
Notice that, according to (\ref{bamu}), $\tau _{c}^{2}>\tau _{c}^{1}$.
\end{remark}

\subsection{Epidemic Persistence in the Steady State}

The steady state is studied by letting the time derivatives $\dot{p}_{i}$
and $\dot{q}_{i}$ equal to zero, i.e.,%
\begin{align}
0& =\beta _{0}(1-p_{i}^{ss}-q_{i}^{ss})\sum_{j\in \mathcal{N}%
_{i}}a_{ij}p_{j}^{ss}  \notag \\
& \qquad \qquad +\beta _{a}q_{i}^{ss}\sum_{j\in \mathcal{N}%
_{i}}a_{ij}p_{j}^{ss}-\delta p_{i}^{ss},  \label{pss} \\
0& =\kappa (1-p_{i}^{ss}-q_{i}^{ss})\sum_{j\in \mathcal{N}%
_{i}}a_{ij}p_{j}^{ss}-\beta _{a}q_{i}^{ss}\sum_{j\in \mathcal{N}%
_{i}}a_{ij}p_{j}^{ss}.  \label{qss}
\end{align}

From (\ref{qss}), it is inferred that%
\begin{equation}
q_{i}^{ss}=\frac{1-p_{i}^{ss}}{1+\frac{\beta _{a}}{\kappa }}\text{ or }\sum
a_{ij}p_{j}^{ss}=0.  \label{qss2}
\end{equation}%
Equivalently, according to (\ref{qss2}), the following is true%
\begin{equation}
q_{i}^{ss}\sum a_{ij}p_{j}^{ss}=\frac{1-p_{i}^{ss}}{1+\frac{\beta _{a}}{%
\kappa }}\sum a_{ij}p_{j}^{ss}.  \label{qss3}
\end{equation}

Now, substitute for $q_{i}^{ss}\sum a_{ij}p_{j}^{ss}$ terms in (\ref{pss})
using (\ref{qss3}) to get%
\begin{align}
& \beta _{0}\frac{\frac{\beta _{a}}{\kappa }}{1+\frac{\beta _{a}}{\kappa }}%
(1-p_{i}^{ss})\sum a_{ij}p_{j}^{ss}  \notag \\
& \qquad +\beta _{a}\frac{1-p_{i}^{ss}}{1+\frac{\beta _{a}}{\kappa }}\sum
a_{ij}p_{j}^{ss}-\delta p_{i}^{ss}=  \notag \\
& \left( \beta _{0}\frac{\frac{\beta _{a}}{\kappa }}{1+\frac{\beta _{a}}{%
\kappa }}+\beta _{a}\frac{1}{1+\frac{\beta _{a}}{\kappa }}\right)
(1-p_{i}^{ss})\sum a_{ij}p_{j}^{ss}-\delta p_{i}^{ss}=0.  \label{pss2}
\end{align}

\begin{theorem}
\label{Theorem: SS}Consider the SAIS epidemic model (\ref{dp}) and (\ref{dq}%
). The steady state values of the infection probabilities of each individual
in the SAIS model is similar to those of the N-intertwined SIS epidemic
model (\ref{NInt_Model}) with a reduced infection rate $\beta _{eq}$.
\end{theorem}

\begin{proof}
Based on the definition of $\beta _{eq}$ in (\ref{bet_eq}), the equation (%
\ref{pss2}) is simplified to 
\begin{equation*}
\beta _{eq}(1-p_{i}^{ss})\sum a_{ij}p_{j}^{ss}-\delta p_{i}^{ss}=0,
\end{equation*}%
which can be expressed as%
\begin{equation}
\frac{\beta _{eq}}{\delta }\sum a_{ij}p_{j}^{ss}=\frac{p_{i}^{ss}}{%
1-p_{i}^{ss}}.  \label{pss3}
\end{equation}

Comparing (\ref{pss3}) with (\ref{NI_pss}) from the Proposition \ref{Prop:
NI_SS}, it is observed that the steady state values of the infection
probabilities in an SAIS epidemic network is similar to those of a SIS
epidemic network with reduced infection rate $\beta _{eq}.$
\end{proof}

\begin{remark}
The expression (\ref{bet_eq}) for $\beta _{eq}$ can be rewritten as%
\begin{equation}
\beta _{eq}=\beta _{0}-\frac{\beta _{0}-\beta _{a}}{1+\frac{\beta _{a}}{%
\kappa }}.
\end{equation}%
The above expression suggests that $\beta _{eq}$ is always less than $\beta
_{0}$ since $\beta _{a}<\beta _{0}$. It is insightful to look at the extreme
cases for the values of $\beta _{eq}$. Particularly, when the alerting rate $%
\kappa $ is very small, $\beta _{eq}\rightarrow \beta _{0}$, indicating that
alertness plays a trivial role in the epidemic spread dynamics. When the
alerting rate is very large, the reduced infection rate $\beta
_{eq}\rightarrow \beta _{a}$. Another case, which is more important from the
epidemiology point of view, is that if $\beta _{a}$ is very small, the
epidemic spread can be completely controlled.
\end{remark}

\section{Simulation Results\label{Sec-Simulation Results}}

Three examples are provided in this section. In all of the simulations, the
curing rate is fixed at $\delta =1$ so that the dimensionless time $\bar{t}%
=\delta t$ is the same as the simulation time.

\begin{example}
\label{Ex1}Consider a contact graph as represented in Fig. \ref%
{example1graph.emf}. For this network, the spectral radius is found to be $%
\rho (A)=3.1385$. The alerting rate is arbitrarily selected as $\kappa =0.1$%
. The infection rate of an alert individual $\beta _{a}$ is chosen $\beta
_{a}=0.1$. For the simulation purpose, nodes $1$, $5$, and $10$ are
initially in the infected state. Other nodes are initialized in the
susceptible state. In each simulation, the total infection fraction $\bar{p}%
(t)=\frac{1}{N}\sum_{i=1}^{N}p_{i}(t)$ is computed. In Fig. \ref%
{example1_eq.eps}, three trajectories are plotted. The trajectory (a)
corresponds to the N-intertwined SIS\ model, with $\beta _{0}=2$. Trajectory
(b) is the solution of the SAIS model (\ref{dp}) and (\ref{dq}) developed in
Section \ref{Sec-Problem Statement}. Trajectory (c) is the solution of the
SIS model but with the reduced infection rate $\beta _{eq}$ defined in (\ref%
{bet_eq}). As is expected from Theorem \ref{Theorem: upperbound}, the
infected fraction in the SAIS\ model is always less than that of the SIS\
model. In addition, as proved in Theorem \ref{Theorem: SS}, the steady state
infection fraction in the SAIS in equal to that of the SIS model with the
reduced infection rate $\beta _{eq}$.
\end{example}

\begin{figure}[h!]
\begin{center}
\leavevmode
\includegraphics[scale=0.40, trim=0cm 0cm 0cm 0cm, clip]{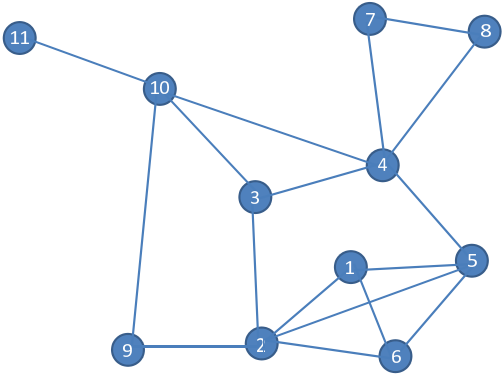}
\end{center}
\caption{The contact graph in Example ~\ref{Ex1} and Example ~\ref{Ex2}.}
\label{example1graph.emf}
\end{figure}

\begin{figure}[h!]
\includegraphics[scale=0.45]{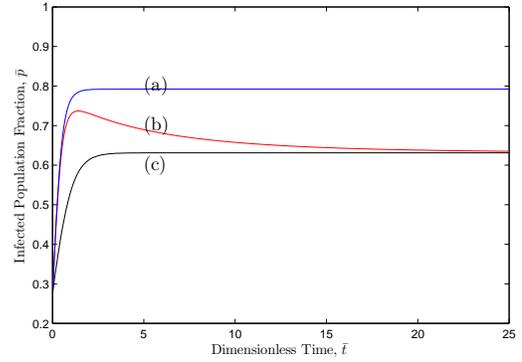}
\caption{The infected population fraction in Example ~\ref{Ex1}. (a) SIS model. (b) SAIS model. (c) SIS model with reduced infection rate $\beta_{eq}$.}
\label{example1_eq.eps}
\end{figure}

\begin{example}
\label{Ex2}In Fig. \ref{example1_eq.eps}, it can be observed that in the
SAIS model the infection spreads similar to the SIS model at the first
stage. Then, the size of the epidemics is reduced due to increased alertness
in the network. In this example, for the same network in the previous
example, the steady state value of the infected fraction and the maximum
value of the infected fraction are presented as a function of the infection
strength $\tau =\beta _{0}/\delta $. The simulation parameters are chosen as 
$\kappa =1,$ $\beta _{a}=0.1$. Note that $\beta _{a}/\delta =0.1<1/\rho
(A)=0.3186$. Therefore, as discussed in Remark \ref{two_thresh}, there
exists two distinct thresholds $\tau _{c}^{1}$ and $\tau _{c}^{2}$ presented
in (\ref{taw_1}) and (\ref{taw_2}), respectively. Simulation results for
this example are shown in Fig. \ref{example1.eps}.
\end{example}

\begin{figure}[h!]
\includegraphics[scale=0.45]{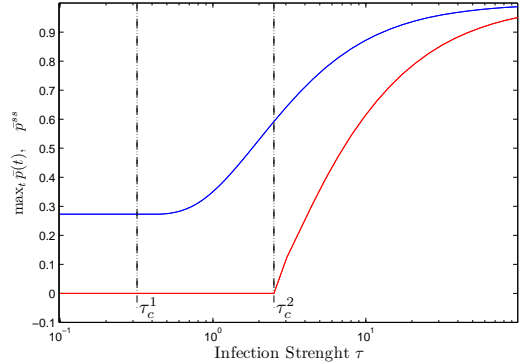}
\caption{The maximum infected fraction (blue line) and the steady state value for the infected fraction (red line) in Example ~\ref{Ex2}.}
\label{example1.eps}
\end{figure}

\begin{example}
As is observed in Fig. \ref{example1.eps}, the steady state values of the
infected fraction $\bar{p}$ is zero before the second epidemic threshold $%
\tau _{c}^{2}$. In addition, the maximum of the infected fraction is equal
to the initial infected fraction before $\tau _{c}^{1}$. The reason for this
observation is that before the first threshold $\tau _{c}^{1}$, the
epidemics dies out exponentially; as stated in Theorem \ref{Theorem: ES}.
Between the two thresholds, $\max_{t}\bar{p}(t)$ is greater than $\bar{p}(0)$
but steady state value $\bar{p}^{ss}=0$. In other words, in this region the
epidemic spreads at the first stage but then is completely controlled as a
result of increased alertness. After the second threshold, $\bar{p}%
^{ss}<\max_{t}\bar{p}(t)$, i.e., alertness reduced the size of the epidemic.
\end{example}

\begin{example}
\label{Ex3}Consider an epidemic network where the contact graph is an
Erdos-Reyni random graph with $N=320$ nodes and connection probability $p=0.2
$. The initial infected population is $\%2$ of the whole population. The
simulation parameters are $\beta _{0}=0.03$, $\kappa =0.05$. Three
trajectories are presented in Fig. \ref{er02n320_1.eps}. The trajectory (a)
is for the SIS model, i.e., no alertness exists. Trajectory (b) is for $%
\beta _{a}=0.02$. In this case, the epidemic size is reduced in the steady
state. Trajectory (c) corresponds to $\beta _{a}=0.01$, for which the
epidemic dies out asymptotically. For the sake of evaluating the model
development in Section \ref{Sec-Problem Statement}, a Monte-Carlo simulation
is also provided for each trajectory and shown in the figure in blue. As can
be seen, there is a reasonable agreement between the proposed model (\ref{dp}%
) and (\ref{dq}) and the Markov process (\ref{Makov1_4}).
\end{example}

\begin{figure}[h!]
\includegraphics[scale=0.45]{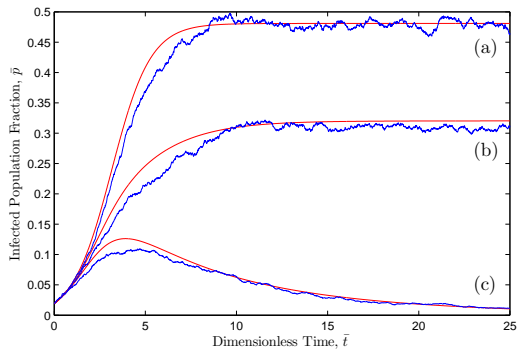}
\caption{The infected population fraction in Example ~\ref{Ex3}. (a) SIS model. (b) SAIS model with $\beta_{a}=0.02$. (c) SAIS model with $\beta_{a}=0.01$. The Monte-Carlo simulation results are shown in blue.}
\label{er02n320_1.eps}
\end{figure}

\section{Acknowledgement}

This research is supported by National Agricultural Biosecurity Center at
Kansas State University. Authors would also like to thank Dr. Fahmida N.
Chowdhury for her constructive feedbacks on this manuscript.

\section{Conclusion}

In this paper, we add a new compartment to the classic SIS model to account
for human response to epidemic spread. Each individual can be infected,
susceptible, or alert. Susceptible individuals can become alert with an
alerting rate if infected individuals exist in their neighborhood. An
individual in the alert state is less probable to become infected than an
individual in the susceptible state; due to a newly adopted cautious
behavior. The problem is formulated as a continuous time Markov process on a
general static graph and then modeled into a set of ordinary differential
equations using mean field approximation method and the corresponding
Kolmogorov forward equations. The model is then studied using results from
algebraic graph theory and center manifold theorem. We analytically show
that our model exhibits two distinct thresholds in the dynamics of epidemic
spread. Below the first threshold, infection dies out exponentially. Beyond
the second threshold, infection persists in the steady state. Between the
two thresholds, the infection spreads at the first stage but then dies out
asymptotically as the result of increased alertness in the network. Finally,
simulations are provided to support our findings. Our results suggest that
alertness can be considered as a strategy of controlling the epidemics which
propose multiple potential areas of applications, from infectious diseases
mitigations to malware impact reduction. Generalizing the current results to
time-varying weighted topologies is a promising extension.

\bibliographystyle{IEEEtran}
\bibliography{AlertEpidemic}

\begin{thebibliography}{10}
\providecommand{\url}[1]{#1}
\csname url@samestyle\endcsname
\providecommand{\newblock}{\relax}
\providecommand{\bibinfo}[2]{#2}
\providecommand{\BIBentrySTDinterwordspacing}{\spaceskip=0pt\relax}
\providecommand{\BIBentryALTinterwordstretchfactor}{4}
\providecommand{\BIBentryALTinterwordspacing}{\spaceskip=\fontdimen2\font plus
\BIBentryALTinterwordstretchfactor\fontdimen3\font minus
  \fontdimen4\font\relax}
\providecommand{\BIBforeignlanguage}[2]{{%
\expandafter\ifx\csname l@#1\endcsname\relax
\typeout{** WARNING: IEEEtran.bst: No hyphenation pattern has been}%
\typeout{** loaded for the language `#1'. Using the pattern for}%
\typeout{** the default language instead.}%
\else
\language=\csname l@#1\endcsname
\fi
#2}}
\providecommand{\BIBdecl}{\relax}
\BIBdecl

\bibitem{ferguson2007Nature}
N.~Ferguson, ``Capturing human behaviour,'' \emph{Nature}, vol. 446, no. 7137,
  p. 733, 2007.

\bibitem{Funk2010JRSI}
S.~Funk, M.~Salath, and V.~A.~A. Jansen, ``Modelling the influence of human
  behaviour on the spread of infectious diseases: a review,'' \emph{Journal of
  The Royal Society Interface}, vol.~7, pp. 1247--1256, 2010.

\bibitem{theodorakopoulos2011ACC}
G.~Theodorakopoulos, J.-Y.~L. Boudec, and J.~S. Baras, ``Selfish response to
  epidemic propagation,'' in \emph{American Control Conference}, 2011, to
  apear.

\bibitem{kitchovitch2010ICCS}
S.~Kitchovitch and P.~Lio, ``Risk perception and disease spread on social
  networks,'' \emph{Procedia Computer Science}, vol.~1, no.~1, pp. 2339--2348,
  2010.

\bibitem{zeng2002JAMIA}
X.~Zeng and M.~Wagner, ``Modeling the effects of epidemics on routinely
  collected data,'' \emph{Journal of the American Medical Informatics
  Association}, vol.~9, no. Suppl 6, p. S17, 2002.

\bibitem{bauch2004NAS}
C.~Bauch and D.~Earn, ``Vaccination and the theory of games,''
  \emph{Proceedings of the National Academy of Sciences of the United States of
  America}, vol. 101, no.~36, pp. 13\,391--4, 2004.

\bibitem{chen2006JMB}
F.~Chen, ``A susceptible-infected epidemic model with voluntary vaccinations,''
  \emph{Journal of mathematical biology}, vol.~53, no.~2, pp. 253--272, 2006.

\bibitem{funk2009NAS}
S.~Funk, E.~Gilad, C.~Watkins, and V.~Jansen, ``The spread of awareness and its
  impact on epidemic outbreaks,'' \emph{Proceedings of the National Academy of
  Sciences}, vol. 106, no.~16, pp. 6872--6877, 2009.

\bibitem{funk2010JTB}
S.~Funk, E.~Gilad, and V.~Jansen, ``Endemic disease, awareness, and local
  behavioural response,'' \emph{Journal of Theoretical Biology}, vol. 264,
  no.~2, pp. 501--509, 2010.

\bibitem{kiss2010MB}
I.~Kiss, J.~Cassell, M.~Recker, and P.~Simon, ``The impact of information
  transmission on epidemic outbreaks,'' \emph{Mathematical biosciences}, vol.
  225, no.~1, pp. 1--10, 2010.

\bibitem{tracht2010PloS}
S.~Tracht, S.~Del~Valle, J.~Hyman, and D.~Carter, ``Mathematical modeling of
  the effectiveness of facemasks in reducing the spread of novel influenza a
  (h1n1),'' \emph{PloS one}, vol.~5, no.~2, p. e9018, 2010.

\bibitem{McKendrick1925EMS}
A.~McKendrick, ``Applications of mathematics to medical problems,''
  \emph{Proceedings of the Edinburgh Mathematical Society}, vol.~44, pp.
  98--130, 1925.

\bibitem{bailey1975Book}
N.~Bailey, \emph{The mathematical theory of infectious diseases and its
  applications}.\hskip 1em plus 0.5em minus 0.4em\relax London, 1975.

\bibitem{Vespignani2002EPJB}
Y.~Moreno, R.~Pastor-Satorras, and A.~Vespignani, ``Epidemic outbreaks in
  complex heterogeneous networks,'' \emph{The European Physical Journal B -
  Condensed Matter and Complex Systems}, vol.~26, pp. 521--529, 2002.

\bibitem{Vespignani2001PRE}
R.~Pastor-Satorras and A.~Vespignani, ``Epidemic dynamics and endemic states in
  complex networks,'' \emph{Phys. Rev. E}, vol.~63, no.~6, p. 066117, May 2001.

\bibitem{wang2003SRDS}
Y.~Wang, D.~Chakrabarti, C.~Wang, and C.~Faloutsos, ``Epidemic spreading in
  real networks: An eigenvalue viewpoint,'' \emph{Proc. 22nd Int. Symp.
  Reliable Distributed Systems (SRDS'03)}, p. 25–34, 2003.

\bibitem{chakrabarti2008TISSEC}
D.~Chakrabarti, Y.~Wang, C.~Wang, J.~Leskovec, and C.~Faloutsos, ``Epidemic
  thresholds in real networks,'' \emph{ACM Transactions on Information and
  System Security (TISSEC)}, vol.~10, no.~4, pp. 1--26, 2008.

\bibitem{Ganesh2005INFOCOM}
A.~Ganesh, L.~Massoulie, and D.~Towsley, ``The effect of network topology on
  the spread of epidemics,'' in \emph{INFOCOM 2005. 24th Annual Joint
  Conference of the IEEE Computer and Communications Societies. Proceedings
  IEEE}, vol.~2, 2005, pp. 1455--1466, 1/rho.

\bibitem{van2009TN}
P.~Van~Mieghem, J.~Omic, and R.~Kooij, ``Virus spread in networks,''
  \emph{Networking, IEEE/ACM Transactions on}, vol.~17, no.~1, pp. 1--14, 2009.

\bibitem{preciado2010CDC}
V.~Preciado and A.~Jadbabaie, ``Spectral analysis of virus spreading in random
  geometric networks,'' in \emph{Decision and Control, Proc. of the 48th IEEE
  Conference on}.\hskip 1em plus 0.5em minus 0.4em\relax IEEE, 2010, pp.
  4802--4807.

\bibitem{preciado2010arXiv}
------, ``Moment-based analysis of spreading processes from network structural
  information,'' \emph{Arxiv preprint arXiv:1011.4324}, 2010.

\bibitem{poletti2009JTB}
P.~Poletti, B.~Caprile, M.~Ajelli, A.~Pugliese, and S.~Merler, ``Spontaneous
  behavioural changes in response to epidemics,'' \emph{Journal of Theoretical
  Biology}, vol. 260, no.~1, pp. 31--40, 2009.

\bibitem{Diestel97Book}
R.~Diestel, ``Graph theory, volume 173 of graduate texts in mathematics,''
  \emph{Springer, Heidelberg}, vol.~91, p.~92, 2005.

\bibitem{Carr1981Book}
J.~Carr, \emph{Applications of Center Manifold Theory}.\hskip 1em plus 0.5em
  minus 0.4em\relax Springer-Verlag, 1981.

\bibitem{khalil2002Book}
H.~Khalil and J.~Grizzle, \emph{Nonlinear systems}.\hskip 1em plus 0.5em minus
  0.4em\relax Prentice hall Englewood Cliffs, NJ, 2002, vol.~3.

\bibitem{stroock2005Book}
D.~Stroock, \emph{{An introduction to Markov processes}}.\hskip 1em plus 0.5em
  minus 0.4em\relax Springer Verlag, 2005.

\bibitem{Piet2006Book}
P.~Van~Mieghem, \emph{Performance analysis of communications networks and
  systems}.\hskip 1em plus 0.5em minus 0.4em\relax Cambridge Univ Pr, 2006.

\end{thebibliography}

\end{document}